\documentclass[10pt,conference]{IEEEtran}
\usepackage{xcolor,soul,framed} 
\usepackage[noadjust]{cite}

\colorlet{shadecolor}{yellow}

\usepackage[pdftex]{graphicx}
\graphicspath{{../pdf/}{../jpeg/}}
\DeclareGraphicsExtensions{.pdf,.jpeg,.png}
\usepackage[caption=false, font=footnotesize]{subfig}
\usepackage{amssymb,mathtools}
\usepackage{array}
\usepackage{mdwmath}
\usepackage{mdwtab}
\usepackage{eqparbox}
\usepackage{url}
\usepackage[ruled,vlined,linesnumbered]{algorithm2e}
\usepackage{enumitem}
\usepackage[mathcal]{euscript} 
\usepackage{tabularx}
\usepackage{multirow} 
\usepackage{booktabs} 
\usepackage{makecell}
\usepackage{aligned-overset} 
\usepackage{url}
\usepackage{multicol}
\usepackage{balance}
\usepackage{bbm}
\usepackage{amsmath}
\usepackage{amsthm} 
\usepackage{pifont}

\graphicspath{{./Figures/}}

\newtheorem{Theorem}{Theorem}

\renewcommand{\vec}[1]{\boldsymbol{\mathrm{#1}}}

\newcommand{\RIS}{\mathtt{R}}

\newcommand{\user}{\mathtt{D}}

\newcommand{\ap}{\mathtt{AP}}

\newcommand{\pl}{\mathrm{PL}}

\usepackage[hidelinks]{hyperref}
\usepackage{xcolor}

\begin{document}

\title{Reliable IoT Communications in 6G Non-Terrestrial Networks with Dual RIS}
\author{
  \IEEEauthorblockN{Muddasir Rahim and Soumaya Cherkaoui}
  \IEEEauthorblockA{Department of Computer Engineering and Software Engineering, Polytechnique Montreal, Canada} \IEEEauthorblockA{Emails: muddasir.rahim@polymtl.ca, soumaya.cherkaoui@polymtl.ca}
}
\maketitle
\begin{abstract}
The increasing demand for Internet of Things (IoT) applications has accelerated the need for robust resource allocation in sixth-generation (6G) networks. In this paper, we propose a reconfigurable intelligent surface (RIS)-assisted upper mid-band communication framework. To ensure robust connectivity under severe line-of-sight (LoS) blockages, we use a two-tier RIS structure comprising terrestrial RISs (TRISs) and high-altitude platform station (HAPS)-mounted RISs (HRISs). To maximize network sum rate, we formulate a joint beamforming, power allocation, and IoT device association (JBPDA) problem as a mixed-integer nonlinear program (MINLP). The formulated MINLP problem is challenging to solve directly; therefore, we tackle it via a decomposition approach. The zero-forcing (ZF) technique is used to optimize the beamforming matrix, a closed-form expression for power allocation is derived, and a stable matching-based algorithm is proposed for device–RIS association based on achievable data rates. Comprehensive simulations demonstrate that the proposed scheme approaches the performance of exhaustive search (ES) while exhibiting substantially lower complexity, and it consistently outperforms greedy search (GS) and random search (RS) baselines. Moreover, the proposed scheme converges much faster than the ES scheme.
\end{abstract}
\begin{IEEEkeywords}
IoT, RIS, HAPS, matching theory, 6G
\end{IEEEkeywords}
\section{INTRODUCTION}\label{introduction}
\IEEEPARstart{T}{he} fifth-generation (5G) wireless communication network has significantly advanced mobile technology. However, the rapid expansion of Internet of Things (IoT) applications has introduced diverse and demanding requirements for next-generation wireless networks, including ultra-high data rates, low latency, and high reliability. To address these challenges and support emerging data-intensive services such as virtual/augmented reality (VR/AR) and ultra-high-definition video streaming, research efforts in both academia and industry have shifted toward developing sixth-generation (6G) wireless systems~\cite{bjornson2025enabling}. While the Sub-6 GHz band used in 5G offers strong reliability, it suffers from bandwidth scarcity. In contrast, the millimeter-wave (mmWave) spectrum, despite its vast frequency resources, suffers from severe propagation losses and limited coverage. Consequently, the mid-band spectrum, particularly the upper mid-band (UMB) within the 7–15 GHz range, has emerged as a strategic candidate for 6G. Recent allocations by the World Radiocommunication Conference 2023 (WRC-23) have designated up to 700 MHz in the upper 6 GHz band and more than 2 GHz between 7 and 15 GHz for mobile communication services~\cite{ghosh2023world}. This frequency range, also referred to as Frequency Range 3 (FR3), is now recognized as the “Golden Band” for 6G, owing to its balance between bandwidth availability and propagation characteristics~\cite{cui20236g,nokia2024}.

To fully exploit the potential of the UMB, efficient techniques are required to overcome propagation losses and enhance coverage. Several solutions have been proposed, including ultra-massive multiple-input multiple-output (UM-MIMO) and reconfigurable intelligent surfaces (RISs). While UM-MIMO provides significant beamforming benefits, its scalability is constrained by its high power consumption and hardware complexity. In contrast, RISs offer a cost-effective and energy-efficient alternative by enabling reconfigurable radio environments through large arrays of passive reflecting elements~\cite{huang2019reconfigurable}. Most existing RIS-assisted networks focus on deploying terrestrial networks (TNs), where RISs are installed on building facades or other static infrastructure to improve coverage. However, in dense urban environments where line-of-sight (LoS) links are frequently obstructed, relying solely on terrestrial RISs (TRISs) can create coverage blind spots~\cite{rahim2023joint,rahim2024multi}. Therefore, upcoming 6G systems are expected to integrate TNs and non-terrestrial networks (NTNs), including Low Earth Orbit (LEO) satellites, High-Altitude Platform Systems (HAPS), and Unmanned Aerial Vehicles (UAVs), to enable seamless global connectivity \cite{you2022enabling}.

Moreover, natural disasters such as earthquakes, hurricanes, floods, and wildfires have repeatedly demonstrated their devastating impact on societies worldwide, often causing large-scale damage to TNs. In such scenarios, HAPS can serve as a resilient and rapidly deployable alternative, providing reliable connectivity when ground-based infrastructure is damaged or inaccessible~\cite{karaman2025demand}. Consequently, integrating HAPS-mounted RISs (HRISs) into the network has emerged as a promising solution to ensure coverage and service reliability, and to enable flexible reconfiguration under both normal and adverse conditions~\cite{you2022enabling}. However, the introduction of this multi-tier, heterogeneous architecture brings new challenges in resource allocation (RA). Therefore, designing efficient RA strategies for RIS-assisted multi-tier communication in the UMB remains an open and crucial research problem with significant implications for future 6G deployments.

\subsection{Related Works}
The standardization of 6G networks begins in 2025 and is expected to conclude around 2028-2029, paving the way for commercial deployment between 2029 and 2030. These networks are anticipated to operate primarily in the UMB spectrum. In this context, extensive research efforts have been initiated to explore the UMB, focusing on its achievable communication performance, including sum rate, latency, and reliability and potential applications. However, only a limited number of studies have investigated MIMO systems within this frequency range~\cite{heath2024beamsharing,bjornson2025enabling,tian2024mid}. In~\cite{heath2024beamsharing}, the authors proposed a user-pairing and beam-sharing method for near-field and far-field users in a MIMO configuration. Similarly,~\cite{bjornson2025enabling} introduced the concept of gigantic MIMO (gMIMO), defined as configurations with at least 256 antenna ports. This gMIMO represents a further expansion of UM-MIMO and analyzes its potential in UMB communications. Furthermore,~\cite{tian2024mid} integrated extra-large-scale MIMO into mid-band communication systems and developed an analytical framework to evaluate key performance metrics, including spectral efficiency (SE) and outage probability (OP).

In~\cite{mohsan2023irs}, the authors presented a comprehensive survey of RIS-assisted UAV communications, summarizing architectures, channel considerations, and design objectives. Authors in~\cite{abdalla2020uavs} focused on the promise of integrating RIS with UAV platforms for future cellular systems, outlining multiple use cases, key challenges, and open research directions. They discussed related works in the domain of spectrum sharing, physical layer security, giant-site access, and enhanced coverage. Moreover, in~\cite{bui2025joint}, the authors proposed a joint optimization framework for transmit power allocation and RIS phase-shift configuration. This framework aimed to maximize energy efficiency while maintaining reliable communication for all users. RISs were deployed on buildings to enhance the coverage and provide a reflective link between the UAVs and users. Meanwhile, the authors in~\cite{farre2025dynamic} introduced integrating NTNs with RIS into existing TNs. This work aimed to improve connectivity and capacity in crowded environments by optimizing coverage and mitigating interference. RIS-assisted UAV communication and HRIS-assisted terrestrial communication are two new schemes proposed by the authors of~\cite{you2022enabling} to jointly apply RIS and UAV for next-generation wireless networks with integrated terrestrial and aerial communications.

The UMBs are the least explored in the context of RIS-assisted communications. To the best of the authors' knowledge, the only existing study highlighting the role of RISs in UMBs is \cite{kara2024reconfigurable}. That work examined various scenarios demonstrating the benefits of RIS deployment and proposed optimal placement strategies to improve UMB system performance. Consequently, further investigation is required to fully harness the potential of integrating HRIS and TRIS in UMBs for future 6G networks.
\subsection{Contributions}\label{contributions}
The studies mentioned above have primarily examined RIS-assisted communication in TNs and have focused mainly on mmWave/terahertz bands. This leaves a gap regarding UMB operation and multi-tier architectures. Motivated by this gap, we investigate integrating TRISs with HRISs for UMB communications. The following are the main contributions of this paper:
\begin{itemize}
    \item To the best of our knowledge, this is the first study that analyzes the impact of TRISs and HRIS in UMBs. We formulate a joint beamforming, power allocation, and device association (JBPDA) problem to maximize the sum rate.
    \item The JBPDA problem is a mixed-integer nonlinear program (MINLP) and is difficult to solve directly. Therefore, we first decompose it into subproblems: beamforming matrix optimization, power allocation, and IoT device association. Then, solve each subproblem efficiently to construct a feasible suboptimal solution.
    \item We use zero-forcing (ZF) beamforming at the access point (AP), derive a closed-form expression for power allocation, and a device-proposing deferred-acceptance matching algorithm for the association sub-problem
    \item The proposed JBPDA achieves sum-rate performance approaching that of exhaustive search (ES) while outperforming the random search (RS) and greedy search (GS) techniques.
\end{itemize}
\section{System Model and Problem Formulation} \label{model}
A multi-layer RIS-assisted IoT network consisting of a single IoT AP equipped with $N$ antennas that serve $K$ single-antenna IoT devices, as shown in Fig.~\ref{m1}. To ensure reliable downlink communication in complex propagation environments, $L$ RISs are deployed across multiple layers, including TRISs and HRISs. The set of all RISs is denoted as $\mathcal{L} = \{\RIS_1, \RIS_2, \ldots, \RIS_l, \ldots, \RIS_L\}$. Each RIS consists of  $M = M_y M_z$ passive reflecting elements, where $M_y$ and $M_z$ represent the number of elements along the horizontal and vertical axes, respectively. Given the severe path loss and frequent blockages at terahertz frequencies, the direct AP-to-device links are assumed to be unavailable. This assumption is reasonable in dense IoT deployments, where physical obstructions and environmental dynamics significantly attenuate direct signal propagation.
\begin{figure}[!t]
     \centering
\includegraphics[width=0.8\linewidth]{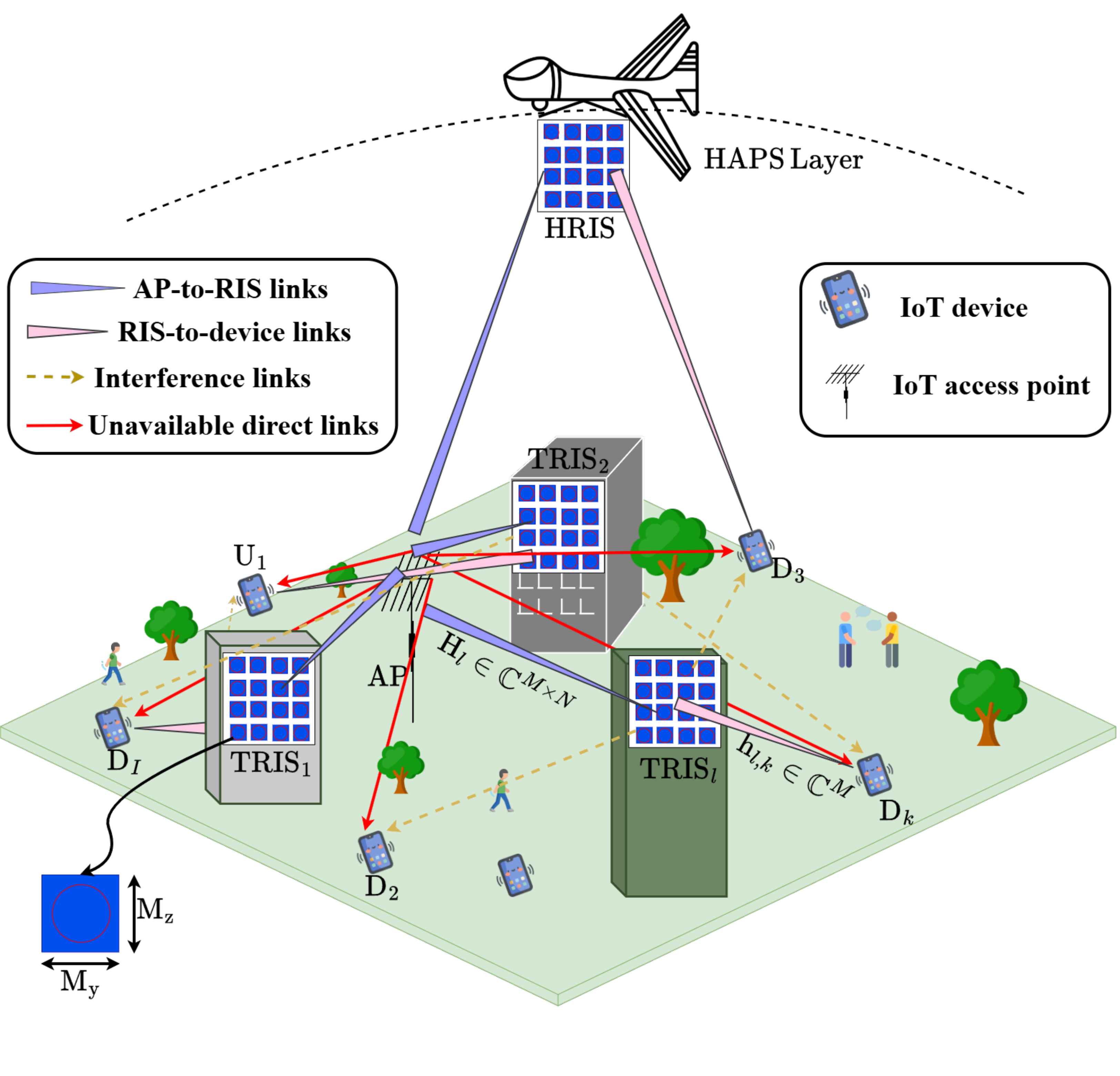}
    \caption{System model of RIS-assisted with terrestrial and HAPS layers}
     \label{m1}   
\end{figure}
In the Cartesian coordinate system, the center locations of the AP, $k^{th}$ IoT device ($\user_{k}$), and $l^{th}$ RIS ($\RIS_{l}$) are represented as ${\mathbf p_{a} = (x_{a}, y_a, z_{a})}$, ${\mathbf p_{k} = (x_{k}, y_{k}, z_{k})}$, and ${\mathbf p_{l} = (x_{l}, y_{l}, z_{l})}$, respectively. Furthermore, the distance from the AP to the center of $\RIS_{l}$ is written as 
\begin{align}
   \textstyle r_{a, l} & 
    =\sqrt{(x_{a}-x_{l})^2+(y_{a}-y_{l})^2+(z_{a}-z_{l})^2}.
\end{align}
Similarly, the distance from $\RIS_{l}$ to $\user_{k}$ is expressed as
\begin{align}
  \textstyle   r_{l, k} &
     =  \sqrt{(x_{l}-x_{k})^2+(y_{l}-y_{k})^2+(z_{l}-z_{k})^2}.
\end{align}
\subsection{Channel Modeling and Data Rate Analysis}
We consider $\vec{H}_l \in \mathbb{C}^{M \times N}$ to represent the channel matrix from the AP to ($\RIS_l$). The channel from the $n^{th}$ antenna of the AP ($\ap_n$) to the $m^{th}$ passive reflecting element of the $\RIS_l$ ($\RIS_{l,m}$) is then given as
 \begin{align}\label{eq1}
  \textstyle  h_{n,(l,m)} = \sqrt{\pl_{n,(l,m)}} \, e^{-j \omega r_{a,l}},
 \end{align}
where $\pl_{n,(l,m)}$ is the pathloss from $\ap_n$ to $\RIS_{l,m}$, $r_{a,l}$ is the distance between the AP and $\RIS_{l}$, and $\omega = \frac{2\pi f}{c}$ is the wave number at frequency $f$. Let us consider the channel between $\RIS_{l}$ and $\user_k$ is $\vec{h}_{l,k}\in \mathbb{C}^{M \times 1}$ and the channel between $\RIS_{l_m}$ and $\user_k$ is given as
\begin{align}
  \textstyle  h_{(l,m),k} = \sqrt{\pl_{(l,m),k}} \, e^{-j \omega r_{l,k}},
\end{align}
where $\pl_{(l,m),k}$ is the pathloss and $r_{l,k}$ is the distance from  $\RIS_{l,m}$ to $\user_k$. Furthermore, the cascaded channel vector from the AP-to-$\RIS_{l}$-to-$\user_k$ can be written as
\begin{align}
\textstyle\vec{g}_{l,k} = \vec{H}_l^{\sf H} \, \vec{\Theta}_l \, \vec{h}_{l,k},
\end{align}
where $\vec{\Theta}_l$ is the RIS configuration matrix, which can be written as
\begin{equation}
 \textstyle   \vec{\Theta}_l = \mathrm{diag}([\kappa_{l,1} e^{j \theta_{l,1}}, \ldots,\kappa_{l,m} e^{j \theta_{l,m}}, \ldots, \kappa_{l,M} e^{j \theta_{l,M}}]),
\end{equation}
where $\kappa_{l,m} \in [0,1]$ and $\theta_{l,m} \in [0, 2\pi)$ represent amplitude coefficients and phase shifts. The channel matrix can be written as 
\begin{equation}
  \textstyle  \vec{G} = [\vec{g}_{1},\ldots,\vec{g}_{k}, \ldots, \vec{g}_{K}]^{\sf H} \in \mathbb{C}^{K \times N}.
\end{equation}
We consider a precoding matrix $\widehat{\vec{W}} \in \mathbb{C}^{N \times K}$ with columns $\widehat{\vec{w}}_k \in \mathbb{C}^{N \times 1}$ denoting unit-norm beam direction for $\user_k$, which can be expressed as $\widehat{\vec{w}}_k = \sqrt{p_k}{\vec{w}}_k,$ where $p_k$ and ${\vec{w}}_k$ are the transmit power scaling factor and the beamforming vector for $\user_k$, respectively. For the given power budget $\vec{P}_{\ap}$, the power constraint can be expressed as
\begin{align}
  \textstyle \sum_{k=1}^K \Vert\widehat{\vec{w}}_k\Vert^2 = \mathsf{Tr} \{ \vec{W}\vec{P} \vec{W}^{\sf H}\} \leq \vec{P}_{\ap},
\end{align}
where $\vec{P} = \mathrm{diag}(p_1,\ldots,p_k,\ldots,p_K)$.

Let $\vec{x} \in \mathbb{C}^{N\times 1}$ be the transmit signal vector from the AP to devices, where
$s_k$ is the data symbol intended for $\user_k$, then, the transmitted signal vector from the AP can be written as 
\begin{align} \label{transmit_signal}
\textstyle\vec{x}=\sum_{{k=1 }}^K\widehat{\vec{w}}_k s_k=\sum_{{k=1  }}^K\sqrt{p_k}{\vec{w}}_k s_k.
\end{align}
The received signals at $\user_k$ via $\RIS_l$ can be expressed as 
\begin{equation}\label{reciver_1}
    y_{l,k} = \underbrace{\sqrt{p_{k}}(\mathbf{g}_{l,k})^H {\vec{w}}_k s_k}_{\text{desired signal}} + \underbrace{\sum_{i \neq k}^{K}\sum_{l=1}^{L}\sqrt{p_{j}}(\mathbf{g}_{l,k})^H {\vec{w}}_i s_i}_{\text{interference}} + \underbrace{n_k}_{\text{noise}},
\end{equation}
where $n_k$ denotes the additive white Gaussian noise (AWGN) with zero mean and variance $\sigma^2$. The first term in \eqref{reciver_1} represents the signal intended for $\user_k$ reflected by the assigned $\RIS_l$ and the second term represents multiuser interference. Furthermore, the signal-to-interference-and-noise ratio (SINR) for $\user_k$ can be written as
\begin{equation}\label{sinr}
    \text{SINR}_{l,k} = \frac{p_{k}|(\mathbf{g}_{l,k})^H {\vec{w}}_k|^2}{\sum_{i \neq k}^{K}\sum_{l=1}^{L}p_j|(\mathbf{g}_{l,k})^H {\vec{w}}_i|^2 + \sigma^2}.
\end{equation}
 Based on the above SINR, the achievable data rate for $\user_k$ can be calculated as
\begin{align}\label{rate}
    R_{l,k} &= \log_2(1 + \text{SINR}_{l,k})\nonumber \\ &=\log_2(1 + \frac{p_{k}|(\mathbf{g}_{l,k})^H {\vec{w}}_k|^2}{\sum_{i \neq k}^{K}\sum_{l=1}^{L}p_j|(\mathbf{g}_{l,k})^H {\vec{w}}_i|^2 + \sigma^2}).
\end{align}

\subsection{Optimization Problem Formulation}\label{OP}
The primary aim of this study is to maximize the network sum rate by optimizing the beamforming matrix, power allocation, and IoT device association. Let $\vec{{W}}$ be the beamforming matrix, $\vec{{P}}$ be the power allocation matrix, and $[\vec{\Upsilon}]_{l,k}$ be the IoT device association matrix. The optimization problem can be written as
\begin{subequations}\label{eq_opt_prob}
\begin{alignat}{2}
& \vec{P} \text{ : } \underset{ \vec{{W}}\,\vec{P}\, \vec{\Upsilon}}{\text{ maximize}}
&\quad
&\textstyle\sum_{k=1}^K R_{l,k}, \label{eq_optProb}\\
&\quad\text{subject to} 
 &&p_k\geq0, \quad\forall k\in K, \label{eq_constraint1}\\
 &&&
\textstyle  \sum_{k=1}^K p_k\Vert{w}_{k}\Vert^2\leq P_{\ap},\label{eq_constraint2}\\
 &&& \textstyle  \sum_{l=1}^L [\vec{\Upsilon}]_{l,k}=1, \quad\forall k\in K , \label{eq_constraint3}\\
&&&\textstyle \sum_{k=1}^K [\vec{\Upsilon}]_{l,k}=1, \quad\forall l\in L , \label{eq_constraint4}\\
&&&  [\vec{\Upsilon}]_{l,k} \in \{0,1\},\quad \forall k,l,\label{eq_constraint5}
\end{alignat}
\end{subequations}
where~\eqref{eq_constraint1} guarantees nonnegative transmit powers and \eqref{eq_constraint2} limits the total power to the AP power budget $P_{\ap}$. Moreover, constraint~\eqref{eq_constraint3} restricts each RIS to a single device, constraint~\eqref{eq_constraint4} assigns each device to exactly one RIS, and constraint~\eqref{eq_constraint5} enforces binary assignment to complete the one-to-one matching.
\section{Proposed Resource Allocation Solutions}\label{proposed}
In this section, we first divide the original optimization problem $\vec{P}$ in \eqref{eq_opt_prob} into sub-problems and solve each sub-problem iteratively. The optimization problem $\vec{P}$ is a MINLP and NP-hard problem. Therefore, we decompose the optimization problem into sub-problems: (i) beamforming-matrix optimization with fixed power and association matrix, (ii) power allocation matrix for the optimized beamforming matrix with fixed association matrix, and (iii) device-RIS association with optimized other parameters. By alternately and optimally solving these subproblems, we obtain a computationally tractable, suboptimal solution to the original optimization problem.
\subsection{Zero Forcing Beamforming}
First, we reformulate the optimization problem to optimize the beamforming matrix for a given power allocation and device-RIS association matrices, as follows: 
\begin{align}\label{P1}
&\vec{P1} \text{ : }  \underset{\vec{{W}}}{\text{ maximize}}
&\sum_{k=1}^K R_{l,k},
&\quad\text{subject to} \,\eqref{eq_constraint1},\eqref{eq_constraint2}.
\end{align}
We derive the beamforming matrix $\vec{W}$ for $\vec{P1}$ in \eqref{P1} using ZF beamforming. ZF is a linear beamforming technique commonly used in multidevice communication networks to eliminate multidevice interference. The derived ZF beamforming matrix using the pseudo-inverse of the channel matrix can be written as
\begin{align} \label{w_zf}
		& \vec{W}=\vec{G}^\dagger = \vec{G}^H(\vec{G}^H\vec{G})^{-1},
		\end{align}
where $\vec{G}^\dagger$ represents the pseudo-inverse of channel $\vec{G}$. Moreover, based on the definition of the ZF beamforming, we have
\begin{align} \label{WH_zf}
		& \vec{G}\vec{W}=\vec{G}\vec{G}^H(\vec{G}^H\vec{G})^{-1} =\vec{I}_K.
\end{align}
Since the interference term is eliminated by ZF beamforming, the received signal and SINR at $\user_k$ via $\RIS_l$ can be respectively written as
\begin{equation}\label{resiver_zf}
    y_{l,k}^{ZF}= \underbrace{\sqrt{p_{k}}(\mathbf{g}_{l,k})^H {\vec{w}}_k s_k}_{\text{desired signal}}+ \underbrace{n_k}_{\text{noise}},
\end{equation}
\begin{equation}\label{sinr_zf}
    \text{SINR}_{l,k}^{ZF} = \frac{p_{k}|(\mathbf{g}_{l,k})^H {\vec{w}}_k|^2}{ \sigma^2}.
\end{equation}
\subsection{Power Allocation}
Using the beamforming matrix achieved from $\vec{P1}$ and the given association matrix, the optimization problem can be reformulated for power allocation as 
\begin{align}\label{P2}
&\vec{P2} \text{ : }  \underset{\vec{{P}}}{\text{ maximize}}
&\sum_{k=1}^K R_{l,k},
&\quad\text{subject to} \,\eqref{eq_constraint1},\eqref{eq_constraint2}.
\end{align}
Considering \eqref{rate} and \eqref{sinr_zf}, we will rearrange $\vec{P2}$ in \eqref{P2} as
\begin{align}\label{P2.1}
\vec{P2} \text{ : }  \underset{\vec{{P}}}{\text{ maximize}}
\quad
&\textstyle\sum_{k=1}^K \log_2\bigg(1 + \frac{p_{k}|(\mathbf{g}_{l,k})^H {\vec{w}}_k|^2}{ \sigma^2}\bigg),\nonumber\\
&\,\text{subject to} \,\eqref{eq_constraint1},\eqref{eq_constraint2}.
\end{align}
Now, the optimization problem P2 in \eqref{P2.1} is a convex problem. Therefore, the optimal power allocation matrix can be derived using the Karush-Kuhn-Tucker (KKT) conditions, as stated in the following theorem.
\begin{Theorem}\label{T2}
The closed-form equation for the power allocation problem $\vec{P2}$ in \eqref{P2.1} can be written as
\begin{align}\label{power}
 \textstyle   P_k=\max\bigg[\frac{\sigma^2}{|(\mathbf{g}_{l,k})^H \mathbf{w}_k|^2}+\frac{1}{\mu + \varrho_k} \bigg]^{+},
\end{align}
where $[x]^+ = \max \{x,0\}$, and $\mu$, and $\varrho_k$ represent the  Lagrangian multipliers. The multiplier $\mu$ is chosen to ensure that $\sum_{k \in K}  p_{k}\leq P_{\ap}$ and $\varrho_k$ ensures the non-negative condition $p_{k}\geq 0$.
\end{Theorem}
\begin{IEEEproof}\label{proof2}
    The Lagrangian function associated with $\vec{P2}$ in~\eqref{P2.1} can be
written as
\begin{align}\label{L1}
\textstyle\mathcal{L} &=\textstyle \sum_{k=1}^{K} \log_2\bigg(1 + \frac{p_k |(\mathbf{g}_{l,k})^H \mathbf{w}_k|^2}{\sigma^2} \bigg)\nonumber\\&\textstyle
- \mu \big( \sum_{k=1}^{K} p_k \Vert{w}_{k}\Vert^2 - P_{\ap} \big)
- \sum_{k=1}^{K} \varrho_k p_k.
\end{align}
Now, let's take the derivative with respect to $p_k$ and set it to zero. We can write using the KKT conditions as
\begin{align}\label{L2}
  \textstyle  \frac{\partial \mathcal{L}}{\partial p_k}
= \frac{|(\mathbf{g}_{l,k})^H \mathbf{w}_k|^2}{\sigma^2 +|(\mathbf{g}_{l,k})^H \mathbf{w}_k|^2 p_k}
- \mu - \varrho_k=0.
\end{align}
Rearrange~\eqref{L2} and apply KKT conditions, we have
\begin{align}\label{L3}
  \textstyle  p_k=\max\bigg[\frac{\sigma^2}{|(\mathbf{g}_{l,k})^H \mathbf{w}_k|^2}+\frac{1}{\mu + \varrho_k} \bigg]^{+}.
\end{align}
\end{IEEEproof}
\subsection{IoT Device-RIS Association}
Then, we use the beamforming matrix and power allocation matrix obtained in subproblems $\vec{P1}$ and $\vec{P2}$, respectively, to formulate the device-RIS association as
\begin{align}\label{P3}
&\vec{P3} \text{ : }  \underset{\vec{\Upsilon}}{\text{ maximize}}
&\quad
&\sum_{k=1}^K R_{l,k},
&\,\text{subject to} \,\eqref{eq_constraint3},\eqref{eq_constraint5}.
\end{align}
\begin{algorithm}[!htp]
  \caption{Proposed IoT Device Association Algorithm for Problem $\vec{P3}$ in~\eqref{P3}}
  \label{algo:1}
  \DontPrintSemicolon{
  \KwIn{ $K$, $L$, $[\vec{\Lambda}]_{\user_k,\RIS_l}$, $[\vec{\Lambda}]_{\RIS_l,\user_k}$, set of unmatched IoT devices $\Pi$, $[\vec{\Upsilon}]_{l,k}$}
 {\bf Initialize }$[\vec{\Lambda}]_{\user_k,\RIS_l} = \emptyset$,
 $[\vec{\Lambda}]_{\RIS_l,\user_k} = \emptyset$, $[\vec{\Upsilon}]_{l,k} = \emptyset$;\;
\underline{\textbf{Preference Matrix:}}\\
\For {$k\in K$}{
$R_{l,k}, \quad \forall\hspace{2mm} l \in L$ and store in $[\vec{R}]_{l,k}$\;
}
$[\vec{R}]_{k,l}= ([\vec{R}]_{l,k})^T$\;
[$[\vec{R}]_{l,k}$,$[\vec{\Lambda}]_{\RIS_l,\user_k}$] = sort ($[\vec{R}]_{l,k}$,2,descend)\;
[$[\vec{R}]_{k,l}$,$[\vec{\Lambda}]_{\user_k,\RIS_l}$] = sort ($[\vec{R}]_{k,l}$,2,descend)\;
\underline{\textbf{IoT Device-RIS Association:}}\\
\While {(either $\Pi$ $\neq \emptyset$ or devices not rejected by all RISs)}{
\For {$\user_{k'}\in \Pi$}{
make a proposal to the RIS with the highest preference in $[\vec{\Lambda}]_{\user_k,\RIS_l}$\;
set element of $[\vec{\Upsilon}]_{l,k} = 1$\;
\For {$\RIS_l\in L$}{
{\bf if} {($\RIS_l \notin [\vec{\Upsilon}]_{l,k}$ )}\;
{
$[\vec{\Upsilon}]_{l,k} \gets [\vec{\Upsilon}]_{l,k} \cup (\user_{k'},\RIS_l) $\; 
}
{\bf else if} {($R_{l,k'}>R_{l,k}$)}\;{
$[\vec{\Upsilon}]_{l,k} \gets [\vec{Psi}]_{l,k} \cup (\user_{k'},\RIS_l)$ \;
{\bf else}\;
$[\vec{\Upsilon}]_{l,k} \gets [\vec{\Upsilon}]_{l,k} $ \;
}
}
  }
{\bf Output: }Association matrix $\vec{\Upsilon}^\star$
}}
\end{algorithm}
We deploy multiple TRISs and a single HRIS to assist the communication links. The device-RIS association $\vec{P3}$ is modeled as a one-to-one matching problem in which each device can be matched to at most one RIS in a slot. Algorithm~\ref{algo:1} implements a device-proposing deferred-acceptance procedure. Initially, all devices are unmatched and maintain preference lists of RISs ordered by their expected data rates. In each round, every unmatched device proposes to the most preferred RIS. Each RIS tentatively accepts the proposal that yields the highest data rate among its current suitor and any existing tentative match and rejects the others. Rejected devices remove that RIS from their lists and may propose to their next choice in subsequent rounds. The process repeats until no further proposals are possible.
\begin{Theorem}\label{T3}
Suppose that Algorithm~\ref{algo:1} is with $K$ IoT devices and $L$ RISs, and that each iteration corresponds to a single proposal event. The proposed algorithm terminates after at most $K\times L$ iterations. Thus, its time complexity is polynomial, denoted as $O(K\times L)$.
\end{Theorem}
\begin{IEEEproof}
Each unmatched IoT device proposes only to RISs it has not yet proposed to, which means the device never proposes to the same RIS twice. Therefore, the total number of distinct proposals across all devices is at most $(K\times L)$. Upon each proposal, the recipient RIS is free to accept the proposal or compare the new proposer with its current pair, keeping the preferred one and rejecting the other, as shown in steps 9 to 20 of Algorithm~\ref{algo:1}. In either case, the set of unproposed RISs for the proposing device strictly shrinks. Because the number of possible proposals is finite ($\le K\times L$), the process must terminate after at most $K\times L$ iterations. At this point, either all IoT devices have been matched or all proposal lists have been exhausted. Hence, the algorithm runs in $O(K\times L)$ time.
\end{IEEEproof}
\section{Simulation Results and Discussion}\label{sim}
We consider a network consisting of multiple IoT devices, multiple RISs, and a single multi-antenna AP deployed over the coverage area. The key simulation parameters are summarized in Table~\ref{tab:sim}. All algorithms are implemented in MATLAB, and unless otherwise specified, each plotted point represents an average over $10^6$ independent channel realizations. The performance of the proposed scheme is compared against three benchmark schemes. First, the ES baseline explores all feasible one-to-one device–RIS associations and selects the association matrix that maximizes the network sum rate. Second, the GS allocation allows each device to select the RIS providing the highest data rate, while the RISs receiving multiple proposals randomly accept one. Third, the RS baseline randomly selects device–RIS associations.
\begin{table}[!t]
\centering
\renewcommand{\arraystretch}{1}
\caption{Simulation parameters}
\label{tab:sim}
\begin{tabular}{l l}
\hline
\textbf{Parameters}                                  & \textbf{Values} \\ \hline
Carrier frequency ($f_c$)  {[}GHz{]}                   & $15$              \\
Number of antennas at AP& $256$\\
AP power budget {[}dBm{]}                    & $23$            \\ 
Power density of noise {[}dBm/Hz{]} & $-174$           \\ 
Number of RIS elements,                               & $100\times100$            \\ 
Channel bandwidth {[}MHz{]}                     & $400$ \cite{karaman2025demand}           \\ 

Noise figure {[}dB{]}                           & $10$ 
\\ 
Side length of RIS elements ${[}m{]}$  & ${\lambda}/{2}$\\ \hline 
\end{tabular}
\end{table}
Fig.~\ref{sum_rate_convergence} illustrates the sum rate versus the number of iterations to demonstrate the convergence behavior of the proposed algorithm. As shown, the proposed scheme exhibits rapid convergence. For instance, with $50$ IoT devices, it converges within approximately $45$ iterations, indicating its high computational efficiency and fast convergence rate. 
\begin{figure}[!htp]
     \centering
\includegraphics[width=0.8\linewidth]{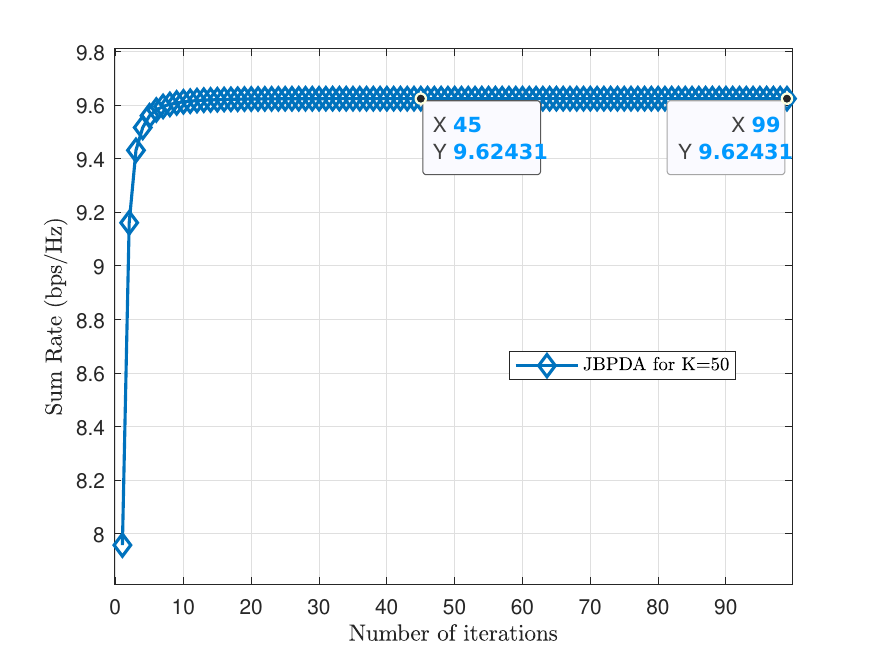}
    \caption{Convergence behavior of the proposed scheme in terms of the sum rate versus the number of iterations for $K=50$ and $M=100\times100$.
     }
     \label{sum_rate_convergence}   
\end{figure}
\begin{figure}[!htp]
     \centering
\includegraphics[width=0.8\linewidth]{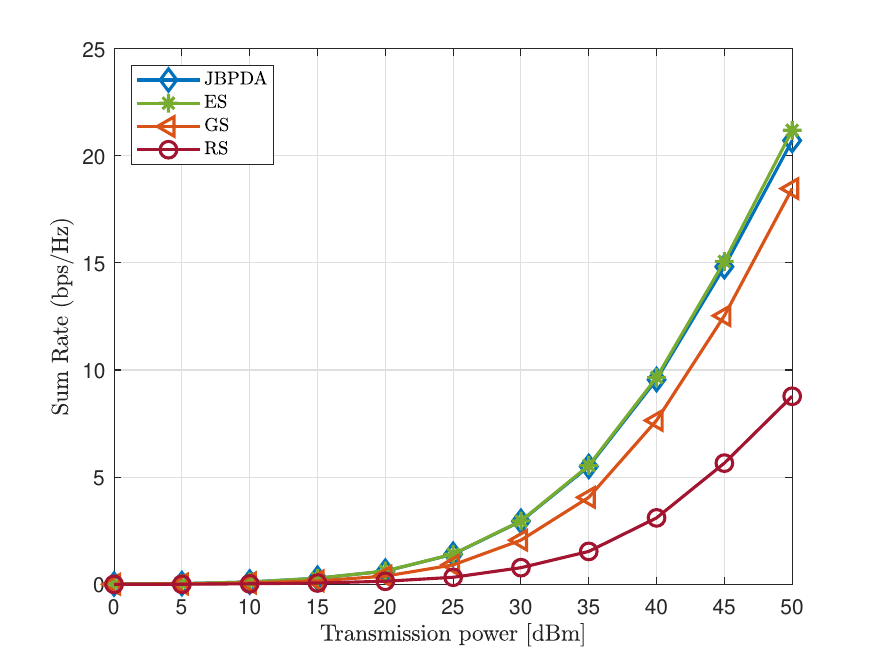}
    \caption{Sum rate versus AP power budget for $K=7$ and $M=100\times100$.}
     \label{sum_rate_power}   
\end{figure}
Fig.~\ref{sum_rate_power} illustrates the sum rate as a function of the transmit power. As expected, all schemes exhibit a monotonically increasing trend. The proposed JBPDA achieves performance within $2\%$ of the ES scheme while consistently outperforming the GS and RS baselines across the entire power range, with sum rates approximately $15\%$ and $57\%$ higher than those of GS and RS, respectively.
Fig.~\ref{sum_rate_power} shows that the sum rate achieved by the proposed JBPDA scheme closely approaches that of the ES benchmark. Therefore, the ES scheme is excluded from Fig.~\ref{sum_rate_users} and Fig.~\ref{sum_rate_ap}, as these scenarios involve more than $100$ IoT devices, making the ES approach computationally infeasible to simulate.
Fig.~\ref{sum_rate_users} illustrates the sum rate as a function of the number of IoT devices. As expected, the overall sum rate increases with the number of devices due to enhanced multiuser diversity and more efficient resource utilization. The proposed JBPDA consistently outperforms the GS and RS baselines across all device counts. Specifically, when the number of IoT devices reaches 200, JBPDA achieves approximately $33\, bps/Hz$ compared to $24\, bps/Hz$ for GS and less than $2 \,bps/Hz$ for RS, confirming its superior scalability and effectiveness in dense IoT scenarios.
\begin{figure}[!t]
     \centering
\includegraphics[width=0.8\linewidth]{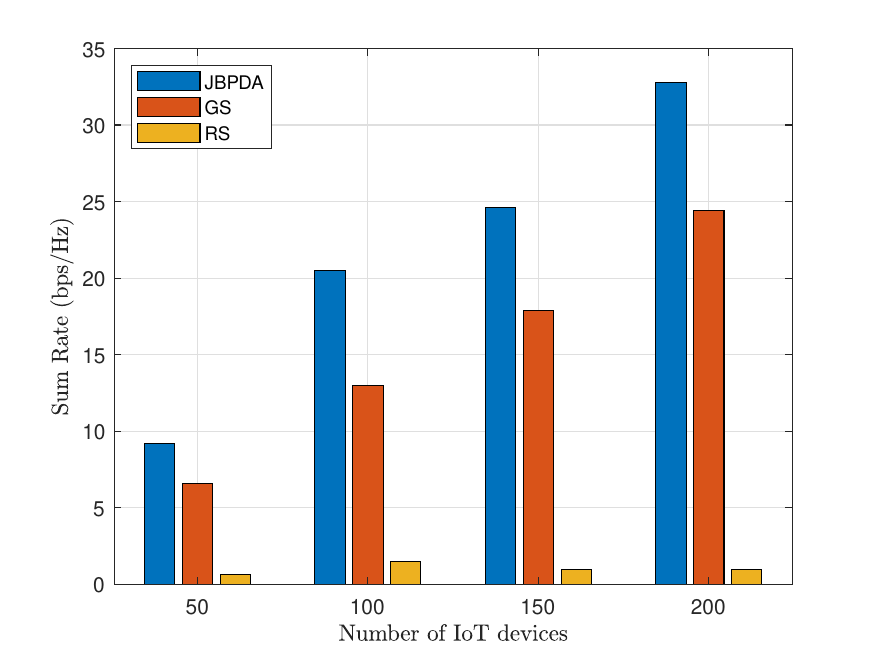}
    \caption{Sum rate versus number of IoT devices for $M=100\times 100$ and $N=256$.
     }
     \label{sum_rate_users}   
\end{figure}
\begin{figure}[!t]
     \centering
\includegraphics[width=0.8\linewidth]{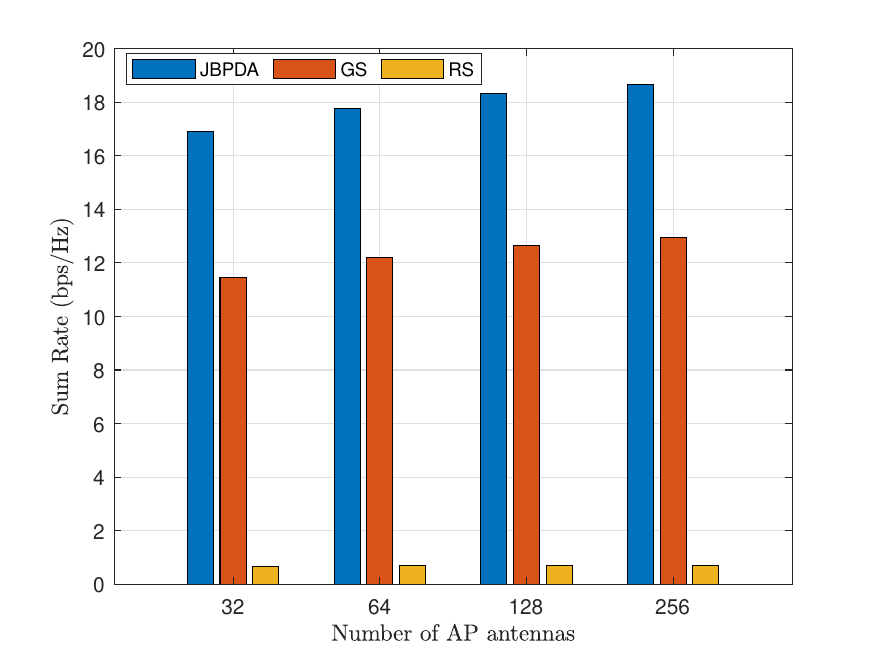}
    \caption{Sum rate versus number of AP antennas for $K=100$ and $M=100\times 100$.
     }
     \label{sum_rate_ap}   
\end{figure}
Fig.~\ref{sum_rate_ap} illustrates the sum rate as a function of the number of antennas at the AP. As expected, the performance of all schemes improves monotonically with increasing antenna count. Across the entire range, the proposed JBPDA consistently outperforms the GS and RS baselines. Specifically, for AP comprising $256$ antennas, the proposed JBPDA achieves approximately $31\%$ and $92\%$ higher sum rates than GS and RS, respectively.
\section{Conclusions}\label{conc}
This paper investigated robust resource allocation for futuristic 6G networks operating in the UMB. A multi-layer RIS-assisted communication framework is proposed to mitigate severe LoS blockages via a two-tier structure comprising TRISs and HRISs. To maximize network performance, a JBPDA problem is formulated. ZF beamforming is used at the AP, a closed-form power-allocation expression is derived, and a stable-matching-based algorithm is developed to associate devices and RISs based on achievable data rates. Numerical results confirmed that the proposed approach significantly improves the network sum rate compared with baseline schemes. In particular, the JBPDA solution achieves near-optimal performance, approaching that of the exhaustive search method while maintaining substantially lower computational complexity. These results validate the effectiveness of the proposed multi-layer RIS architecture for reliable 6G deployments. Future extensions of this work include integrating a LEO satellite segment to enhance adaptability, coverage, and scalability of the hierarchical RIS design. Additionally, hybrid beamforming at the AP and across RISs under IoT device mobility will be investigated, accounting for tracking and Doppler effects, and incorporating handover to represent practical deployment scenarios better.
\bibliographystyle{IEEEtran}
\balance
\bibliography{References}

\end{document}